\documentclass[]{tMOP2e}

\newcommand{\bra}[1]{\langle #1|}
\newcommand{\ket}[1]{| #1 \rangle }
\newcommand{\ip}[2]{{\langle #1|}{ #2 \rangle }}
\newcommand{\tr}[1]{{\rm tr}[#1]}
\newcommand{\be}{\begin{eqnarray}}
\newcommand{\ee}{\end{eqnarray}}

\newcommand{\cE}{{\cal E}}
\newcommand{\cI}{{\cal I}}

\newcommand{\cF}{{\cal F}}
\newcommand{\cT}{{\cal T}}
\newcommand{\cS}{{\cal S}}
\newcommand{\cH}{{\cal H}}

\newcommand{\cR}{{\cal R}}

\newcommand{\cQ}{{\cal Q}}

\begin{document}

\issn{1362-3044 }
\issnp{0950-0340} \jvol{57} \jnum{03} \jyear{2010} \jmonth{10 January}

\markboth{Taylor \& Francis and I.T. Consultant}{Journal of Modern Optics}


\title{Single-shot discrimination of quantum unitary processes}
\author{
M\'ario Ziman$^{a,b,\ast}$\thanks{$^\ast$Email:ziman@savba.sk},
and Michal Sedl\'ak$^{a}$
\\\vspace{6pt}
$^{a}${\em Research Center for Quantum Information, Institute of Physics, Slovak Academy of Sciences, D\'ubravsk\'a cesta 9, 845 11 Bratislava, Slovakia}
\\\vspace{6pt}
$^{b}${\em Faculty of Informatics, Masaryk University, Botanick\'a 68a, 602 00 Brno, Czech Republic}
\\\vspace{6pt}
\received{today 2009}
}

\maketitle

\begin{abstract}
We formulate minimum-error and unambiguous
discrimination problems for quantum processes in the language
of process positive operator valued measures (PPOVM). In this framework
we present the known solution for minimum-error discrimination
of unitary channels. We derive
a ``fidelity-like'' lower bound on the failure probability of
the unambiguous discrimination of arbitrary quantum processes. This bound
is saturated (in a certain range of apriori probabilities) in the case of
unambiguous discrimination of unitary channels. Surprisingly, the optimal
solution for both
tasks is based on the optimization of the same quantity called
completely bounded process fidelity.
\bigskip
\bigskip
\end{abstract}


\section{Introduction}
Quantum Theory is intrinsically a statistical theory,
which means that our predictions and conclusions are typically
probabilistic (see for example \cite{helstrom,peres_book}).
For instance, even having the best possible knowledge
on the photon polarization and polarizer filter we cannot predict whether
an individual photon will pass the polarizer, or not. For us, as observers,
this event is random except for very specific cases.
Consequently, the predictive abilities of Quantum Theory are necessarily
formulated in the language of probabilities.

On the other hand, in experiments we do not meet directly with probabilities.
If the statistical samples are sufficiently large to estimate
the probabilities, our conclusions about the identities of quantum objects
could have a deterministic flavor. The remaining uncertainties are related
to potential incompleteness of the information contained
in the measured probabilities. For example, a measurement
of the $z$th component of the spin (by means of Stern-Gerlach experiment)
does not tell us almost anything about the
$x$th coordinate of the spin. However, after sufficiently many (infinitely)
repetitions the $z$th component is determined perfectly without any
uncertainty.

In this paper we shall focus on our ability to make
conclusions based on measurements repeated at most finite (small)
number of times. Our primary aim is to investigate the distinguishability of
quantum channels having access only to limited number of tests.
We shall be interested in two particular statistical tasks:
minimum-error discrimination and unambiguous discrimination. Both of them
were extensively studied in the case of states, however,
the discrimination of quantum processes is still rather
an unexplored research area. In particular,
researchers investigated the minimum error
distinguishability of unitary channels \cite{acin,dariano_unitary}.
Partial results were obtained also in the unambiguous discrimination
\cite{chefles,wang} and minimum-error discrimination of specific channels
\cite{sacchi_1,sacchi_2,watrous,sacchi_3,li,piani}.

This paper is structured as follows: In Sections I, II, and III we will
introduce the necessary concepts and mathematical tools.
The case of state discrimination
is very briefly discussed in Section IV. The Section V presents
the general framework for discrimination of channels and the discrimination
of unitary channels is analyzed in details in Section VI.

\section{Description of experiments}
An experiment is a time ordered set of instructions that are
divided into three procedures: i) preparation, ii) processing,
iii) measurement. In quantum theory the quantum systems are associated
with Hilbert spaces and the mathematical description of quantum objects
(preparators, processes and measurements) is formulated in terms of specific
operators and structures defined on the underlying Hilbert space $\cH$.

The goal of preparations is to design a source of systems
in particular quantum states, which are represented by density operators, i.e.
positive linear operators of a unit trace. Let us denote by $\cS(\cH)$
the set of quantum states, i.e.
$\cS(\cH)=\{\varrho:\varrho\ge O,\tr{\varrho}=1\}$.
The events observed in the performed measurement are described by
positive operators $O\leq E\leq I$ called also effects. Let us note that
the positivity $A\ge O$ means that $\ip{\psi}{A\psi}\ge 0$
for all $\psi\in\cH$ and
$A\ge B$ is equivalent to positivity of $A-B\ge O$.
The probability
to observe an effect $E$ providing that the measured state was $\varrho$
is given by the relation $p=\tr{\varrho E}$. The whole measurement
is described by a collection of effects $E_1,\dots,E_n$
associated with $n$ mutually exclusive events forming a so-called
positive operator valued measure (POVM), i.e. the normalization
$\sum_j E_j=I$ holds. Thus, the observed
probability distribution of outcomes $E_1,\dots,E_n$ reads
$p_j=\tr{\varrho E_j}$.

In some cases it is convenient to include the processing part
into either the preparation, or the measurement. However, in this paper
the processes will be tested in experiments and therefore
we shall consider them as devices independent of preparators and measurements.
Mathematically, the processes are modeled as channels, i.e. completely positive
trace-preserving linear maps defined on the set of
trace-class operators $\cT(\cH)(\supset\cS(\cH))$.
In particular, a linear map $\cE:\cT(\cH)\to\cT(\cH)$ is completely
positive, if $\cI\otimes\cE[\Omega_+]\ge O$, where $\cI$ denotes the identity
map, $\Omega_+=\sum_{jk}
\ket{\varphi_j\otimes\varphi_j}\bra{\varphi_k\otimes\varphi_k}$
and $\varphi_1,\dots,\varphi_d$ is an orthonormal basis of $\cH$.
It is trace-preserving if $\tr{\cE[X]}=\tr{X}$ for all trace-class
operators $X\in\cT(\cH)$.

\section{Classes of discrimination problems}
In the discrimination problem the goal is to design an experiment in
which an unknown quantum device (preparator, process, measurement)
is used only once (or finite number of times) and from the
observed outcome (sequence of outcomes) we want to determine
which of $N$ expected
elements ``best'' fits as the description of
the unknown device. Let us denote by $X=\{x_j\}_{j\in J}$
the possible outcomes and let $\Omega=\{\omega_1,\dots,\omega_N\}$
be the set of $N$ conclusions. The set $\Omega$ plays a dual role. It
also represents the apriori information on the identity of the
discriminated object in a sense that
we know that the unknown device is one of the elements in $\Omega$.
The conditional probability $p(x_j|\omega_k)$ gives the probability
to observe an outcome $x_j$ providing that the device is actually
described by $\omega_k$. Defining the apriori distribution
$\eta:\Omega\to [0,1]$ and using the Bayes rule we get the
conditional probability
\be
p(\omega_k|x_j)=\frac{\eta_k p(x_j|\omega_k)}{\sum_l \eta_l p(x_j|\omega_l)}
\ee
evaluating the reliability of the conclusion $\omega_k$ providing
that the outcome $x_j$ is observed. Let us note that
$p_j=\sum_l \eta_l p(x_j|\omega_l)$ is the total probability to observe
the outcome $x_j$. If $p(\omega_k|x_j)=1$ for some $\omega_k$, then
the outcome $x_j$ uniquely determines conclusion $\omega_k$. We shall
call such outcome and the related conclusion {\it unambiguous}.
In all other cases,
the conclusions are necessarily erroneous. In particular,
$1-p(\omega_k|x_j)$ is the related conditional error probability,
when we choose the conclusion $\omega_k$ for the outcome $x_j$.

We can formulate many different discrimination problems. In what follows
we shall consider two variations: minimum-error discrimination and
unambiguous discrimination. In the so-called minimum-error discrimination
\cite{helstrom}, the goal is to minimize on average the errors we made
in our conclusions. For simplicity, let us assume that the outcome $x_j$ leads
to conclusion $\omega_j$. Then the average error reads
\be
\overline{p}_{\rm error}=1-\sum_{j\in J} \eta_j p(x_j|\omega_j)\,.
\ee
In the unambiguous discrimination problem the goal is either to achieve
an unambiguous conclusion, or do not make any conclusion
\cite{ivanovic,dieks,peres}. Therefore, the conclusions,
if made, are error-free. However, not making any conclusion
results in a nonvanishing {\it failure probability}, which on average
reads
\be
\overline{p}_{\rm fail}=\sum_{k\in J_{\rm inc}}\sum_{\omega_j\in\Omega}\eta_j p(x_k|\omega_j)\, ,
\ee
where $J_{\rm inc}$ denotes the set of indices associated with inconclusive
outcomes. The aim is to minimize this quantity while satisfying the
unambiguity of conclusions.

\section{Discrimination of states}
Discrimination problems for quantum states were investigated
from many different perspectives, but in some versions the complete solutions
are still not known. Let us briefly mention the basic results in the
minimum-error discrimination of a preparator, which is known to
produce one of the states $\varrho_1$, $\varrho_2$
with apriori probabilities $\eta_1,\eta_2$, respectively.
The statistics of the most general experiment we can perform might be
 formulated in the language of POVM, i.e.
a pair of positive operators $E_1,E_2$ such that $E_1+E_2=I$.
That is, $\Omega=\{\varrho_1,\varrho_2\}$ and $X=\{E_1,E_2\}$, where
outcome associated with $E_j$ is used to conclude $\varrho_j$.
Since the probabilities are given by the formula
$p(E_j|\varrho_k)=\tr{E_j\varrho_k}$ we get \cite{helstrom}
\be
\nonumber
\overline{p}_{\rm error}&=&
\min_{\rm POVM}\left(1-\eta_1\tr{E_1\varrho_1}-\eta_2\tr{E_2\varrho_2}\right)\\
\nonumber
&=&\frac{1}{2}\min_{\rm POVM}\left(1-\tr{(E_1-E_2)(\eta_1\varrho_1-\eta_2\varrho_2)}\right)
\\
\nonumber
&=& \frac{1}{2}(1-{\rm tr}|\eta_1\varrho_1-\eta_2\varrho_2|)
\\
&=& \frac{1}{2}(1-||\eta_1\varrho_1-\eta_2\varrho_2||_{\rm tr})
\,,
\ee
where $||\cdot||_{\rm tr}={\rm tr}|\cdot|$ is the {\it trace norm}.
The minimum is achieved for $E_1=\Pi_+$, where $\Pi_+$ is a projector
onto the eigenvectors of the operator
$\Delta=\eta_1\varrho_1-\eta_2\varrho_2$ associated with the positive
eigenvalues.

Unlike the minimum-error discrimination the unambiguous one
does not have a nontrivial solution for a general pair of states
$\varrho_1,\varrho_2$. There are cases in which the unambiguity
requirements $\tr{E_1\varrho_2}=\tr{E_2\varrho_1}=0$
cannot be satisfied. In the unambiguous discrimination we are looking for
effects $E_1,E_2$ such that $E_1+E_2\leq I$ and an effect
$I-E_1-E_2$ represents the inconclusive outcome.
In particular, the unambiguous discrimination is possible only if
the supports of $\varrho_1$ and $\varrho_2$ do not coincide.
Interestingly, if $\varrho_1$, $\varrho_2$ are apriori
equally probable pure states $\psi, \varphi$,
then $\overline{p}_{\rm fail}=|\ip{\psi}{\varphi}|$
(see for example \cite{ivanovic,dieks,peres,jaeger,rudolph}).
Although many interesting results have been discovered
 \cite{rudolph,raynal,feng,herzog,kleinmann}, we are
lacking a closed formula for the optimal value of $\overline{p}_{\rm fail}$
in the general situation.

\section{Discrimination of channels}
In this section we shall formulate analogous discrimination problems
for quantum processes, i.e. channels. A general experiment
for probing them is described
by the so-called process POVM in the same sense as POVM describes general
experiment measuring the properties of quantum states. Process POVM
provides a compact representation of the statistics generated by the
most general experimental setup probing the properties of quantum
channels.

The framework of PPOVM exploits a specific representation of channels
defined via so-called Choi-Jamiolkowski isomorphism
\cite{pillis,jamiolkowski,choi}.
According to this theorem a channel
on $d$ dimensional system can be represented by a positive operator acting
on $d\times d$ system. In particular, a channel $\cE$ is represented
by an operator
$\Omega_\cE=(\cI\otimes\cE)[\Omega_+]$, where
$\Omega_+=\sum_{j,k} \ket{\varphi_j\otimes\varphi_j}
\bra{\varphi_k\otimes\varphi_k}$. Let us note that $\Omega_+$ is not
a projector, because it is not normalized and $\tr{\Omega_+}=d$.
The operator $\frac{1}{d}\Omega_+$ is a one-dimensional projector onto
the maximally entangled state $\psi_+=\frac{1}{\sqrt{d}}
\sum_j \varphi_j\otimes\varphi_j\in\cH\otimes\cH$.

Process POVM is defined \cite{ziman_ppovm,dariano} as a collection
of positive operators (effects)
$M_1,\dots,M_n$ such that $\sum_j M_j = \xi^T\otimes I$ for some
state $\xi\in\cS(\cH)$. An event that can be observed
in the experiment consists of
a preparation of the test state $\varrho$
and an observation of the effect $E_j$ in the measurement $E$
of the output state.
Let us note that in the experiment we are allowed to use an ancilla
of arbitrary size, i.e. $\varrho$ and $E_j$ are operators defined on
$d_{\rm anc}\times d$-dimensional Hilbert space. The conditioned
probability to observe an event consisting of
the state preparation $\varrho$ and the observation of an effect $E_j$
providing that channel $\cE$ is tested equals
\be
p(\varrho,E_j|\cE)=\tr{E_j(\cI\otimes\cE)[\varrho]}\, .
\ee
Using the Choi-Jamiokowski relation
$\varrho=(\cR_\varrho\otimes\cI)[\Omega_+]$, where
$\cR_\varrho:\cT(\cH)\to\cT(\cH_{\rm anc})$ is a completely
positive map, and the duality relation
$\tr{Y\cF[X]}=\tr{\cF^*[Y] X}$ determining the dual channel $\cF^*$
we can write
\be
\nonumber
p(\varrho,E_j|\cE)&=&\tr{(\cR_\varrho^*\otimes\cI)[E_j](\cI\otimes\cE)[\Omega_+]}\\
&=&\tr{M_j\Omega_\cE}\,,
\ee
where $M_j$ is an element of PPOVM. By definition $M_j$ is positive and
$\sum_j M_j=(\cR^*_\varrho\otimes\cI)[I]=\xi^T\otimes I$, where
$\xi={\rm tr}_{\rm anc}[\varrho]$. Thus, any experiment in which the
channel is used once can be formalized as a PPOVM and the converse
also holds \cite{ziman_ppovm}, i.e. any PPOVM can be experimentally implemented.

\subsection{Minimum-error discrimination}
The framework of PPOVM is very useful for the formulation of the discrimination
problems, because we do not have to consider all the details related to
preparation of the test states and measurements.
Let us formulate the minimum-error discrimination problem for a pair
of channels $\cE_1,\cE_2$ represented by operators $\Omega_1,\Omega_2$.
Analogously as in the case of states
the aim is to design a PPOVM (given by $M_1,M_2$)
minimizing the error probability
\be
\nonumber
\overline{p}_{\rm error}=
\frac{1}{2}\min_{\rm PPOVM}\left(1-\tr{(M_1-M_2)(\eta_1\Omega_1-\eta_2\Omega_2)}\right)\,,
\ee
where $M_1+M_2=\xi^T\otimes I$ for some state $\xi$. Although this formula
is similar to the one for the state discrimination, the optimization is
due to the freedom in the normalization of the PPOVM more complex
and not yet sufficiently understood. In fact,
$\xi^T$ cannot be $I$, because $\tr{\xi^T}=1$. Therefore,
the optimization for channels does not reduce to an optimization for states.
For instance,
pure states can be perfectly distinguished only if
they are orthogonal, however, for unitary channels the orthogonality
(with respect to the Hilbert-Schmidt scalar product)
is only a sufficient condition \cite{acin,dariano_unitary}.

For every PPOVM there exists a pure test state realization, i.e.
$M_j=\cR_\psi^*\otimes\cI[F_j]$ for some pure test state represented
by a unit vector $\psi\in\cH\otimes\cH$
and $\{F_1,F_2\}$ is a POVM defined in $\cH\otimes\cH$ system.
Expressing PPOVM elements in this way
we obtain a well-known formula for the minimum error
probability (see for example \cite{raginsky2002,sacchi_2})
\be
\nonumber
\overline{p}_{\rm error}&=&\frac{1}{2}-\frac{1}{2}\sup_{\psi,F_1,F_2}
\tr{(F_1-F_2)(\cI\otimes(\eta_1\cE_1-\eta_2\cE_2))[P_\psi]}\\
\nonumber
&=&\frac{1}{2}-\frac{1}{2}\sup_\psi
{\rm tr}|(\cI\otimes(\eta_1\cE_1-\eta_2\cE_2))[P_\psi]| \\
&=& \frac{1}{2}(1-||\eta_1\cE_1-\eta_2\cE_2||_{\rm cb})
\ee
where $||\cdot||_{\rm cb}$ is the so-called {\it norm of complete boundedness}
\cite{paulsen96} and $P_\psi=\ket{\psi}\bra{\psi}$.

A simple upper bound on this probability is given by an experiment
in which the maximally entangled state {\bf $\psi_+$} is used as the test state, i.e.
$M_j=\frac{1}{d}F_j$, where $F_j$ are effects forming the performed
POVM, hence $M_1+M_2=\frac{1}{d}I\otimes I$. The bound reads
\be
\overline{p}_{\rm error}\leq \frac{1}{2}(1-\frac{1}{d}{\rm tr}
|\eta_1 \Omega_1-\eta_2\Omega_2|)\, .
\ee
Another interesting bound comes from the experiments in which no ancilla
is used, i.e. $M_j=\ket{\psi}\bra{\psi}^T\otimes F_j$, where
$F_j$ is the POVM measurement of the output state. In such case
\be
\overline{p}_{\rm error}\leq
\frac{1}{2}(1-\max_{\psi\in\cH} ||(\eta_1\cE_1-\eta_2\cE_2)[P_\psi]||_{\rm tr})\,.
\ee

\subsection{Unambiguous discrimination}
In the case of the unambiguous discrimination the problem is formulated
by means of the following equations
\be
&\tr{M_1\Omega_2}=\tr{M_2\Omega_1}=0&\\
&\overline{p}_{\rm failure}=\min_{M_0} \tr{M_0(\eta_1\Omega_1+\eta_2\Omega_2)}&
\ee
under the PPOVM constraint
\be
M_0+M_1+M_2=\xi^T\otimes I
\ee
for some state $\xi\in\cS(\cH)$.

In the following proposition we shall formulate a lower bound
on the probability of failure, which is analogous to the bound
known for the unambiguous discrimination of two mixed states
(see for instance \cite{feng}).
\begin{proposition}
\label{prop:1}
Let $\cE_1,\cE_2$ be channels and $\eta_1,\eta_2$
be their apriori probabilities. Then
\be
\overline{p}_{\rm failure}\ge 2\sqrt{\eta_1\eta_2}\min_{\xi\in\cS(\cH)}
{\rm tr}|\sqrt{\Omega_1}(\xi^T\otimes I)\sqrt{\Omega_2}| \, ,
\ee
where $\Omega_j=(\cI\otimes\cE_j)[\Omega_+]$.
\end{proposition}
\begin{proof}
Since for all numbers $a^2+b^2\ge 2ab$ and setting $a=\eta_1\tr{M_0\Omega_1}$,
$b=\eta_2\tr{M_0\Omega_2}$ we get
\be
\overline{p}^2_{\rm failure}\ge 4\eta_1\eta_2 \tr{M_0\Omega_1}\tr{M_0\Omega_2}\, .
\ee
Using the Cauchy-Schwartz inequality we obtain
\be
\nonumber
& & \tr{M_0\Omega_1}\tr{M_0\Omega_2}=
\\\nonumber & &
\quad=\tr{U\Omega_1^{\frac{1}{2}}
M_0^{\frac{1}{2}}M_0^{\frac{1}{2}}\Omega_1^{\frac{1}{2}}U^\dagger}
\tr{\Omega_2^{\frac{1}{2}}M_0^{\frac{1}{2}}M_0^{\frac{1}{2}}\Omega_2^{\frac{1}{2}}}\\
\nonumber & &\quad\geq (\tr{U\sqrt{\Omega_1}M_0\sqrt{\Omega_2}})^2\, .
\ee
By definition $M_0=\xi^T\otimes I-M_1-M_2$. Since the no-error
conditions $\Omega_1 M_2=M_1\Omega_2=O$ hold, it follows that
$\sqrt{\Omega_1}M_0\sqrt{\Omega_2}=\sqrt{\Omega_1}
(\xi^T\otimes I)\sqrt{\Omega_2}$, thus,
\be
\overline{p}_{\rm failure}\geq
2\sqrt{\eta_1\eta_2}|\tr{U\sqrt{\Omega_1}(\xi^T\otimes I)\sqrt{\Omega_2}}|\,.
\ee
Using the identity $\sup_U |\tr{XU}|={\rm tr}|X|$ holding for all
operators $X$ the inequality reads
\be
\overline{p}_{\rm failure}\geq 2\sqrt{\eta_1\eta_2}
{\rm tr}|\sqrt{\Omega_1}(\xi^T\otimes I)\sqrt{\Omega_2}|\,,
\ee
which proves the lemma after the optimalization over the PPOVM
normalization is taken into account.
\end{proof}

The function $D(\Omega_1,\Omega_2)=\min_\xi
{\rm tr}|\sqrt{\Omega_1}(\xi\otimes I)\sqrt{\Omega_2}|$
we shall call {\it completely bounded process fidelity}
in analogy with the completely bounded norm $||\cdot||_{\rm cb}$.
Let us note that both $\xi$ and $\xi^T$ are states, thus
the transposition is irrelevant in the formula for $D$.
This quantity was introduced in Ref.\cite{belavkin2005} under the name
{\it minimax fidelity} as the abstract channel analogy of the state fidelity.
Since \cite{belavkin2005}
\be
1-\frac{1}{2}||\cE_1-\cE_2||_{\rm cb}
\leq D(\Omega_1,\Omega_2)\leq \sqrt{1-\frac{1}{4}||\cE_1-\cE_2||_{\rm cb}^2}\,,
\ee
we get $2\overline{p}_{\rm error}\leq D(\Omega_1,\Omega_2)$
for $\eta_1=\eta_2=1/2$. Consequently, if the identity
$D(\cE_1,\cE_2)=\min_\xi{\rm tr}|\sqrt{\Omega_1}(\xi\otimes I)\sqrt{\Omega_2}|
=0$
holds, the channels $\cE_1,\cE_2$ can be
perfectly discriminated. Equivalently, the condition
\be
\Omega_1(\xi\otimes I)\Omega_2=O
\ee
(holding for some density operator $\xi$)
implies that the channels represented by $\Omega_1,\Omega_2$
are perfectly distinguishable, and vice versa \cite{chiribella2008}.

\section{Unitary channels}
In this section we shall focus on the discrimination of a pair
of unitary channels. The minimum-error approach was investigated
in \cite{acin,dariano_unitary} and the unambiguous approach was
adopted by Chefles et al.
in \cite{chefles}. Unitary channels are associated with Choi-Jamiokowski
operators proportional to one-dimensional projectors. In particular,
$\cE_U=U\cdot U^\dagger$ is represented by $\Omega_U=d\ket{\psi_U}\bra{\psi_U}$,
where $\psi_U=(I\otimes U)\psi_+$. Given a pair of unitary channels $U,V$, then
the joint support of $\Omega_U,\Omega_V$ specifies a two-dimensional
subspace $\cQ$ of $\cH\otimes\cH$, which is relevant
for both discrimination problems.

\subsection{Minimum-error approach}
Evaluation of the cb-norm $||\eta_U \cE_U-\eta_V\cE_V||_{\rm cb}$
will give us the solution for the minimum-error discrimination.
Each unit vector $\psi$ can be expressed as $\psi=(A\otimes I)\psi_+$,
thus, $P_\psi=(\cR_\psi\otimes\cI)[\Omega_+]
=\frac{1}{d}(A\otimes I)\Omega_+(A^\dagger\otimes I)$. Moreover, since
the following identity holds for any pair of vectors
$\psi,\varphi$ and apriori probabilities $\eta_\psi,\eta_\varphi$
\be
\nonumber
{\rm tr}
|\eta_\psi\ket{\psi}\bra{\psi}-\eta_\varphi\ket{\varphi}\bra{\varphi}|
=\sqrt{1-4\eta_\psi\eta_\varphi|\ip{\psi}{\varphi}|^2}
\ee
we get \cite{dariano_unitary} the formula
\be
\nonumber
\overline{p}_{\rm error}&=&
\frac{1}{2}(1-||\eta_U \cE_U-\eta_V\cE_V||_{\rm cb})\\
&=&
\frac{1}{2}(1-\sqrt{1-4\eta_U\eta_V D^2})\,,
\ee
where
\be
\nonumber
D&=&
\min_{A: \tr{A^\dagger A}=d}
|\ip{(A\otimes U)\psi_+}{(A\otimes V)\psi_+}|\\
\nonumber &=&
\frac{1}{d}\min_{A}|\tr{(A^\dagger A)^TU^\dagger V}|
\\
\label{eq:D} &=&
\min_{\xi\in\cS(\cH)} |\tr{\xi U^\dagger V}|\,.
\ee
We used the identities $(A\otimes I)\psi_+=(I\otimes A^T)\psi_+$ and
$dA^\dagger A={\rm tr}_{\rm anc}\ket{(A\otimes I)\psi_+}\bra{(A\otimes I)\psi_+}
=\xi^T$, where $\xi$ denotes the reduced state of the subsystem
entering the tested quantum channel.

\subsection{Unambiguous approach}
Since supports of $\Omega_U$  and $\Omega_V$ are different,
two unitaries can be always unambiguously distinguished.
Let us denote by $Q$ a projector onto the linear subspace $\cQ$
spanned by vectors $\psi_U,\psi_V$. The unambiguous no-error conditions
require that on the relevant subspace $\cQ$ the operators
$M_U,M_V$ are rank-one and take the form
\be
M_U^\cQ&=&c_U(Q-\ket{\psi_V}\bra{\psi_V})\,,\\
M_V^\cQ&=&c_V(Q-\ket{\psi_U}\bra{\psi_U})\,.
\ee
In addition, $M_U+M_V\leq\xi^T\otimes I$ for some state $\xi$.
The success probability $\overline{p}_{\rm success}=1-\overline{p}_{\rm failure}$
reads
\be
\nonumber
& & \overline{p}_{\rm success}=\max_{\rm PPOVM}\left(
\eta_U\tr{M_U\Omega_U}+\eta_V\tr{M_V\Omega_V}\right)\\
\nonumber
& &\quad=\max_{\rm PPOVM}\left(
\eta_U\tr{M_U^\cQ\Omega_U}+\eta_V\tr{M_V^\cQ\Omega_V}\right)\\
\nonumber
& &\quad=
\max_{\varphi}\max_{\rm POVM}\left(
\bra{\varphi_U}\eta_U F_U\ket{\varphi_U}+
\bra{\varphi_V}\eta_V F_V\ket{\varphi_V}\right)
\ee
As previously, we used the fact that PPOVM can be always implemented 
using a pure test state. This test state is associated with a suitable
vector $\varphi=(A\otimes I)\psi_+$
leading to
$M_U=(A^\dagger\otimes I)F_U(A\otimes I)$,
$M_V=(A^\dagger\otimes I)F_V(A\otimes I)$,
where effects $F_U,F_V$ represent the conclusive
outcomes of the performed POVM, i.e. $F_U+F_V\leq I\otimes I$.
We used the notation $\varphi_U=(I\otimes U)\varphi$ and
$\varphi_V=(I\otimes V)\varphi$.

For a fixed test state $\ket{\varphi}\bra{\varphi}$ the POVM
maximizing the expression
$\bra{\varphi_U}\eta_U F_U\ket{\varphi_U}+
\bra{\varphi_V}\eta_V F_V\ket{\varphi_V}$
is known from the analogous problem of unambiguous pure
state discrimination \cite{jaeger,rudolph}. Without loss of generality
we can assume that $\eta_U\ge\eta_V$. In such case the optimal
POVM consists of effects
\be
\nonumber
F_U=\min\left\{\frac{1-\sqrt{\frac{\eta_V}{\eta_U}}|\ip{\varphi_U}{\varphi_V}|}{1-|\ip{\varphi_U}{\varphi_V}|^2},1\right\}
(Q_\varphi-\ket{\varphi_V}\bra{\varphi_V})\,,\\
\nonumber
F_V=\max\left\{\frac{1-\sqrt{\frac{\eta_U}{\eta_V}}|\ip{\varphi_U}{\varphi_V}|}{1-|\ip{\varphi_U}{\varphi_V}|^2},0\right\}(Q_\varphi-\ket{\varphi_U}\bra{\varphi_U})\,,
\ee
where $Q_\varphi$ is a projector onto the subspace spanned by
vectors $\varphi_U,\varphi_V$. The failure probability
$\overline{p}_{\rm failure}=1-\overline{p}_{\rm success}$
reads
\begin{equation}
\nonumber
\overline{p}_{\rm failure}=
\left\{
\begin{array}{ll}
2\sqrt{\eta_U\eta_V}D\quad &
{\rm if}\  D\leq \sqrt{\frac{\eta_V}{\eta_U}}\leq 1
\\
\eta_V+\eta_U D^2 &
{\rm if}\  D\geq \sqrt{\frac{\eta_V}{\eta_U}}\leq 1
\end{array}
\right.\,,
\end{equation}
where we used the definition $D=\min_{\varphi}|\ip{\varphi_U}{\varphi_V}|
=\min_{\xi\in\cS(\cH)}|\tr{\xi^TU^\dagger V}|$ coinciding with Eq.(\ref{eq:D}).

Let us note that the considered unambiguous discrimination of unitary
channels saturates the bound specified in Proposition
\ref{prop:1} for values $D\leq \sqrt{\frac{\eta_V}{\eta_U}}$.
Indeed, since ${\rm tr}|X|={\rm tr}\sqrt{X^\dagger X}$ and
$\sqrt{\Omega_U}=\sqrt{d}\ket{\psi_U}\bra{\psi_U}$,
$\sqrt{\Omega_V}=\sqrt{d}\ket{\psi_V}\bra{\psi_V}$
the bound gives
\be
\nonumber
\overline{p}_{\rm failure}&\ge& 2\sqrt{\eta_U\eta_V}\min_\xi
{\rm tr}|\sqrt{\Omega_U}(\xi^T\otimes I)\sqrt{\Omega_V}|\\
\nonumber
&=& 2d\sqrt{\eta_U\eta_V}\min_\xi
|\ip{\psi_U}{(\xi^T\otimes I)\psi_V}| \tr{\ket{\psi_U}\bra{\psi_U}}\\
\nonumber &=&
2\sqrt{\eta_U\eta_V}\min_\xi |\tr{\xi^TU^\dagger V}|\\
\nonumber &=& 2\sqrt{\eta_U\eta_V}D\,,
\ee
where we used the identity
$\ip{\psi_U}{(\xi^T\otimes I)\psi_V}=\frac{1}{d}\tr{\xi^TU^\dagger V}$.
If $D\geq\sqrt{\frac{\eta_V}{\eta_U}}$, then
$\overline{p}_{\rm failure}=\eta_V+\eta_UD^2>2\sqrt{\eta_U\eta_V}D$, because
$(\sqrt{\eta_V}-D\sqrt{\eta_U})^2\ge 0$. We see that this bound is not
achievable in general. In fact, the existence of the
PPOVM giving the bound is not guaranteed in its derivation.
The particular process discrimination problem could pose additional
constraints on the possible choices of the normalization $\xi^T\otimes I$,
which makes the value of $D$, hence also the bound, different.

\subsection{Evaluation of $D$}
It follows that the optimal solutions of both discrimination
problems for unitary channels is based on minimalization of the same
quantity $D$, which is called completely bounded process fidelity.
This quantity was also analyzed in the study of perfect
discrimination of unitary channels \cite{acin,dariano_unitary}
and we will repeat the analysis.
Let us denote by $\{\phi_k\}$ the eigenvectors of $U^\dagger V$
associated with eigenvalues $e^{i\theta_k}$. Then
\be
D=\min_{\xi\in\cS(\cH)} |\sum_k e^{i\theta_k} \bra{\phi_k}\xi\ket{\phi_k}|\,.
\label{ddefhull}
\ee
The number on the right hand side is a convex combination of complex square
roots of unity. Thus, it can be
visualized as an element of the convex hull of points
(eigenvalues of $U^\dagger V$) on the
unit circle of the complex plane. Our aim is to find the complex number
within this convex hull which is closest to zero. In particular,
if 0 is not contained in the convex hull, then
\be
D=\frac{1}{2}\min_{k,l}|e^{i\theta_k}+e^{i\theta_l}|\, ,
\ee
which means a suitable test state has only two nonvanishing entries (equal to
1/2) on the diagonal of its reduced state $\xi$ (see Figure \ref{fig}).

\begin{figure}
\begin{center}
\includegraphics[width=7cm]{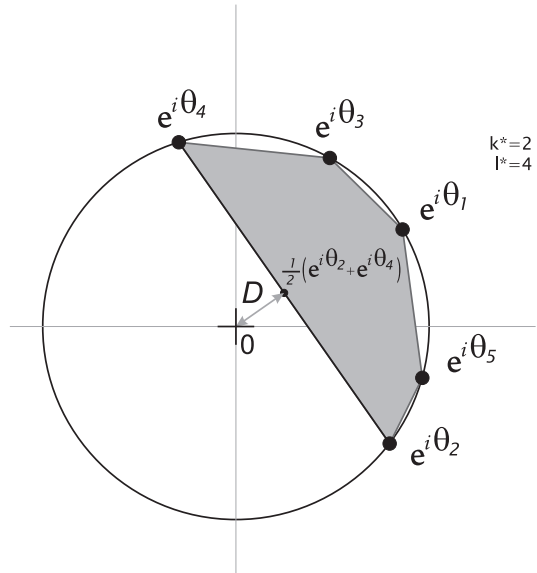}
\caption{Illustration of the completely bounded
process fidelity $D$ for  unitary channels in the case,
when the convex hull of eigenvalues of $U^\dagger V$ does not contain 0.}
\end{center}
\label{fig}
\end{figure}

Since for two-dimensional Hilbert space the unitary operators
have only two eigenvalues, the minimalization is trivial \cite{acin} and
reads
\be
D=\frac{1}{2}|e^{\i\theta_1}+e^{i\theta_2}|=\frac{1}{2}|\tr{U^\dagger V}|\, .
\ee
Hence in this case the orthogonality in the Hilbert-Schmidt sense
is necessary and sufficient for perfect discrimination of $\cE_U$
and $\cE_V$. Moreover, the maximally entangled state (for which
$\xi=\frac{1}{2}I$) is the universal test state optimizing the
minimum-error and unambiguous discrimination. Of course, the
measurements depend on the particular task and the unitaries.
However, these properties do not hold in the higher dimensions.
For example, CNOT and SWAP gate can be perfectly discriminated
even without being orthogonal and the maximally entangled
test state is not very usable.

The minimum in the definition of $D$ (see Eq.(\ref{ddefhull}))
depends only on the diagonal entries of $\xi$, thus we can always
choose optimal $\xi$ to be a pure state.
That is, no ancilla is needed in order to implement
an optimal discrimination experiment.
Formally, the optimal test state
can be chosen to be factorized $\psi=\psi_A\otimes\psi_S$, where
$\psi_A$ is arbitrary and $\psi_S$ is the pure test state
with suitable diagonal elements $|\ip{\phi_k}{\psi_S}|^2$.
Let us assume that $k^*,l^*$ are indexes of the eigenvalues optimizing the
average error probability. Then,
$\psi_S=\frac{1}{\sqrt{2}}(\varphi_{k^*}+\varphi_{l^*})$ is the vector associated
with an optimal test state. For any apriori probabilities $\eta_U, \eta_V$ this single test state is optimal for both minimum error and unambiguous discrimination. The optimal experiments for these tasks differ in the used measurements, which depends also on the apriori probabilities.

\section{Conclusion}
The discrimination of quantum devices provides us with a clear operational
definition of their closeness. Apart from this purely mathematical motivation,
the discrimination problems naturally appear in various communication
and computation problems. In this paper we formulated the minimum-error
and unambiguous single-shot discrimination among two quantum processes
using the language of PPOVM. In this framework we can clearly see the
differences between the discrimination tasks for states
and for processes. Many of the results derived for states can be
translated to channels, however,
there are also some significant differences. As for example, the perfect
distinguishability of pure states and unitary channels
\cite{acin,dariano_unitary}. For the minimum-error approach
the trace norm is replaces by completely bounded norm, which
is not that easy to evaluate in general \cite{rosgen,watrous09}.
We derived a simple lower bound on the probability of failure for unambiguous
discrimination of quantum channels
\be
\overline{p}_{\rm failure}\ge 2\sqrt{\eta_U\eta_V} D(\Omega_1,\Omega_2)\,.
\ee
This bound suggests a state fidelity equivalent for channels called completely
bounded process fidelity
\be
D(\Omega_1,\Omega_2)=\min_\xi{\rm tr}|\sqrt{\Omega_1}
(\xi\otimes I)\sqrt{\Omega_2}|\,,
\ee
where $\Omega_j=(\cI\otimes\cE_j)[\Omega_+]$ are the Choi-Jamiolkowski
operators associated with the channels $\cE_j$. Let us remind that
in the case of minimum-error discrimination the optimal
value of error probability equals
\be
\overline{p}_{\rm error}=\frac{1}{2}(1-||\eta_1\cE_1-\eta_2\cE_2||_{\rm cb})\,.
\ee

For unitary channels we have shown that both discrimination problems
reduce to the optimization of the same quantity.
Moreover, in this case the lower bound on the probability of failure
is saturated. In particular, for equal apriori probabilities
$\eta_U=\eta_V=1/2$ we have
\be
\overline{p}_{\rm error}&=&\frac{1}{2}(1-\sqrt{1-D^2})\,,\\
\overline{p}_{\rm failure}&=&D\,,
\ee
where
\be
D=D(\cE_U,\cE_V)=\min_\xi {\rm tr}|\xi U^\dagger V|\, .
\ee
Interestingly, no ancilla is required and the same pure test states
optimizes both probabilities simultaneously. The difference is
in the measurement performed on the channel output.

A lot of work remains to be done in the area of process discrimination
and identification. We believe that a better understanding to distinguishability
of general quantum processes is related to the development of the
theory of PPOVM, which currently serves as a useful tool for numerical
optimization. However, to get a deeper understanding of discrimination problems
it seems crucial to be able to characterize those PPOVMs that are compatible
with the given constraints.

\subsection*{Acknowledgments}
We acknowledge financial support via the European Union projects QAP
2004-IST-FETPI-15848, HIP FP7-ICT-2007-C-221889, and via the projects
APVV-0673-07 QIAM, VEGA-2/0092/09, OP CE QUTE ITMS NFP 262401022,
and CE-SAS QUTE.


\end{document}